\documentclass{kms-j}

\newcommand{\publname}{J.~Korean Math.~Soc.}
\issueinfo{}% volume number
  {}%        % issue number
  {}%        % month
  {}%     % year
\pagespan{1}{}
\copyrightinfo{}%              % copyright year
  {Korean Mathematical Society}% copyright holder
\usepackage{graphicx}
\usepackage{multirow}
\allowdisplaybreaks

\theoremstyle{plain}
\newtheorem{theorem}{Theorem}[section]

\newtheorem{lemma}[theorem]{Lemma}
\newtheorem{corollary}[theorem]{Corollary}

\theoremstyle{definition}

\theoremstyle{remark}
\newtheorem{remark}[theorem]{Remark}

\begin{document}

\title[Several classes of three-weight or four-weight linear codes ]
{Several classes of three-weight or four-weight linear codes }

\author[Q. Liao]{Qunying Liao}
\address{Qunying Liao \\College of Mathematical Science \\ Sichuan Normal University \\ Chengdu Sichuan, 610066, China}
\email{qunyingliao@sicnu.edu.cn}

\author[Z. Zhang]{Zhaohui Zhang}
\address{Zhaohui Zhang \\College of Mathematical Science \\ Sichuan Normal University \\ Chengdu Sichuan, 610066, China}
\email{1192114967@qq.com}

\author[P. Zheng]{Peipei Zheng}
\address{Peipei Zheng \\College of Mathematical Science \\ Sichuan Normal University \\ Chengdu Sichuan, 610066, China}
\email{18383258130@163.com}

\thanks{This work was financially supported by National Natural Science Foundation of China (Grant No. 12071321 \& 12471494 \& 12461001).}

\subjclass[2020]{Primary: 94B05; Secondary: 94A24.}
\keywords{linear code, defining set, weight distribution, three-weight, four-weight}

\begin{abstract}
In this manuscript, we construct 
a class of  projective three-weight linear codes and two classes of projective four-weight linear codes over $\mathbb{F}_{2}$ from the defining sets construction, and determine their weight distributions by using additive characters. Especially, the  projective three-weight linear code and one class of projective four-weight linear codes (Theorem 4.1) can be applied in secret sharing schemes.
\end{abstract}

\maketitle

\section{Introduction}

	Denote $\mathbb{F}_{p^{m}}$ as the finite field with $p^{m}$ elements and $\mathbb{F}_{p^{m}}^{*}=\mathbb{F}_{p^{m}} \backslash\{0\}$, where $p$ is a prime. An $[n, k, d]$ linear code $\mathcal{C}$ of length $n$ over $\mathbb{F}_{q}$ is a $k$-dimensional linear subspace of $\mathbb{F}_{q}^{n}$ with minimum (Hamming) distance $d$, where $q$ is a power of the  prime $p$. Let $A_{i}$ denote the number of codewords in $\mathcal{C}$ with weight $i$. The sequence $(1, A_{1}, \ldots, A_{n})$ is called the weight distribution of $\mathcal{C}$ and the polynomial $1+A_{1} z+\cdots+A_{n} z^{n}$ is called the weight enumerator of $\mathcal{C}$.  $\mathcal{C}$ is called a $t$-weight code if the number of nonzero $A_{j}$ in the sequence $(A_{1},  \ldots, A_{n})$ is equal to $t$. The complete weight enumerator is an important parameter for a linear code, obviously, the weight distribution can be deduced from the complete weight enumerator. In addition, the weight distribution of $\mathcal{C}$ can be applied to determine the capability for both error-detection and error-correction\cite{A1}. The dual code $\mathcal{C}^{\perp}$ of $\mathcal{C}$ is defined by $\mathcal{C}^{\perp}=\left\{\mathbf{x} \in \mathbb{F}_{q}^{n}: \mathbf{x} \cdot \mathbf{c}=0\right.$ for all $\left.\mathbf{c} \in \mathcal{C}\right\}$. A linear code whose
     dual code has the minimal distance $d^{\perp}\geq3$ is called a projective code. 

    Linear codes over finite fields are applied in data storage devices, computer and communication systems, and so on. In particular, few-weight linear codes have been better applications in secret sharing schemes \cite{A2,A3}, association schemes \cite{A4}, strongly regular graphs\cite{A60}, authentication codes \cite{A5}, and so on. 

   A number of two-weight or three-weight linear codes have been constructed \cite{A13,A15,A18,A19,A20,A22,A23,A24,A25,A26,A28,A30,A31,A32,A33,A35,A36,A38,A39,A41,A44}.
   Especially, in 2015, Ding et al.\cite{A18} proposed a method to construct two-weight or three-weight linear codes from the defining set. As follows, Ding et al. introduced the linear code 
\[
\mathcal{C}_{D}=\left\{(\operatorname{Tr}(x d_{1}), \operatorname{Tr}(x d_{2}), \ldots, \operatorname{Tr}(x d_{n})) | x \in \mathbb{F}_{p^{m}}\right\}
\]
by choosing
$$D=\left\{d_{1}, d_{2}, \ldots, d_{n}\right\} \subseteq \mathbb{F}_{p^{m}}^{*}.$$

In 2023, Zhu et al.\cite{A46} introduced the linear code 
\[
\mathcal{C}_{\tilde{D}}=\left\{(\operatorname{Tr}(a y x^{d}+b x))_{(x, y) \in \tilde{D}} |(a, b) \in \mathbb{F}_{p^{m}} \times \mathbb{F}_{p^{m}}\right\}
\]
by choosing
\[
D^{*}=\left\{(x, y) \in \mathbb{F}_{p^{m}}^{*} \times \mathbb{F}_{p^{m}}^{*} \mid \operatorname{Tr}(y x^{d+1})=0\right\}
\]
and
\[
D_{\lambda}=\left\{(x, y) \in \mathbb{F}_{p^{m}}^{*} \times \mathbb{F}_{p^{m}} \mid \operatorname{Tr}(y x^{d+1})=\lambda\right\},
\]
where $\tilde{D}=D^{*} \text { or } D_{\lambda}$, $\lambda \in \mathbb{F}_{p}$ and $d$ is a positive integer.

In this manuscript, for a given positive integer $d$ and $m>2$, we consider the linear code $\mathcal{C}_{D_{i}}$ which is defined
by
\begin{align}
	&\mathcal{C}_{D_{i}}=\left\{\mathbf{c}(a, b)=(\operatorname{Tr}(a x^dy +bx))_{(x, y) \in D_{i}}: a, b \in \mathbb{F}_{2^{m}}\right\}(i=1,2,3),
\end{align}	
where
\begin{align}
	&D_{1}=\left\{(x, y) \in \mathbb{F}_{2^{m}}^{*} \times \mathbb{F}_{2^{m}}: \operatorname{Tr}(x^{d+1}y+x^{d}y+x)=0\right\},\\
	&D_{2}=\left\{(x, y) \in \mathbb{F}_{2^{m}}^{*} \times \mathbb{F}_{2^{m}}^{*}: \operatorname{Tr}(x^{d+2}y+x^{d+1}y)=0\right\},\\
	&D_{3}=\left\{(x, y) \in \mathbb{F}_{2^{m}}^{*} \times \mathbb{F}_{2^{m}}: \operatorname{Tr}(x^{d+1}y)=0~and~\operatorname{Tr}(x^{d}y)=1\right\}.
\end{align}
The parameters and the complete weight enumerator of $\mathcal{C}_{D_{i}}$ are determined based on additive characters. In particular, $\mathcal{C}_{D_{1}}$ and $\mathcal{C}_{D_{2}}$ are  both suitable for applications in secret sharing schemes.

$\mathcal{C}_{D_{1}}$ for $m$ even and Theorem $3$ in \cite{A57}
have the same weight distribution, therefore we only investigate the weight distribution of $\mathcal{C}_{D_{1}}$ for $m$ odd.

This manuscript is organized as follows. In Section $2$, we provide a brief summary for the relevant properties of the trace function and additive characters over finite fields. In Section $3$, we present some necessary results that will be utilized in  Section 4. In Section $4$, we give our main results and their proofs.  In Section $5$, we conclude the whole manuscript.

\section{Preliminaries}

	In this section, we first introduce the concepts of the trace function and additive characters over finite fields.
\begin{lemma}\cite{A47}
	(a) For $\alpha \in \mathbb{F}_{2^{m}}$, the trace function $\operatorname{Tr}(\alpha)$ of $ \alpha$ over $\mathbb{F}_{2^{m}}$ is defined by
	$$\operatorname{Tr}(\alpha)=\alpha+\alpha^{2}+\cdots+\alpha^{2^{m-1}};$$
	
	(b) when $m$ is odd, $\operatorname{Tr}(1)=1;$
	
	(c) $\operatorname{Tr}(\alpha+\beta)=\operatorname{Tr}(\alpha)+\operatorname{Tr}(\beta)$ for all $\alpha, \beta \in \mathbb{F}_{2^{m}}.$
\end{lemma}

Let ${\zeta_{p}}$ be a primitive ${p}$-th root of unity and ${\zeta_{p}}=\frac{2\pi i}{p}$$(i=\sqrt{-1})$. An additive character ${\chi_{i}}$ of ${\mathbb{F}_{2^{m}}}$ over ${\mathbb{F}_{2}}$ is a homomorphism from ${\mathbb{F}_{2^{m}}}$ into the multiplicative group composed by ${\zeta_{2}}$. For any ${i \in \mathbb{F}_{2^{m}},~\chi_{i}(x)=(-1)^{\operatorname{Tr}(i x)}}$, where ${x \in \mathbb{F}_{2^{m}}}$. If ${i=0, ~\chi_{0}(x)=1}$ for all ${x \in \mathbb{F}_{2^{m}}}$ and ${\chi_{0}}$ is  called the trivial additive character of ${\mathbb{F}_{2^{m}}}$. If ${i=1, ~\chi_{1}}$ is called the canonical additive character of ${\mathbb{F}_{2^{m}}}$, which is recorded as ${\chi}$ in this manuscript. The orthogonality of the additive character is given by
\begin{align}\label{C1}
	&\sum_{x \in \mathbb{F}_{2^{m}}} \chi(x)=\sum_{x \in \mathbb{F}_{2^{m}}}(-1)^{\operatorname{Tr}(x)}=0,
\end{align}
\begin{align}\label{C2}
	&\sum_{x \in \mathbb{F}_{2^{m}}^{*}} \chi(x)=\sum_{x \in \mathbb{F}_{2^{m}}^{*}}(-1)^{\operatorname{Tr}(x)}=-1.
\end{align}

\begin{lemma}\cite{A6}
	Let $\mathcal{C}$ be an ${[n, k]}$ code over ${\mathbb{F}_{2}}$ with weight distribution $(1, A_{1}, \ldots, A_{n})$ and $\mathcal{C^{\perp}}$ be the dual code with weight distribution $(1, A_{1}^{\perp}, \ldots, A_{n}^{\perp})$, then the first three Pless power moments are as follows,
	$$
	\begin{array}{l}
		\displaystyle{\sum_{j=0}^{n}} A_{j}=2^{k}, \\
		\displaystyle{\sum_{j=0}^{n}} j A_{j}=2^{k-1}(n-A_{1}^{\perp}),\\
	
		\displaystyle{\sum_{j=0}^{n}} j^{2} A_{j}=2^{k-2}(n(n+1)-2n A_{1}^{\perp}+2 A_{2}^{\perp}) .
	\end{array}
	$$
\end{lemma}

  For a linear code $\mathcal{C}$ with length $n$, the support of a codeword $\mathbf{c}=(c_{1}, \ldots, c_{n}) \in$ $\mathcal{C} \backslash\{\boldsymbol{0}\}$ is denoted by
\[
\operatorname{supp}(\mathbf{c})=\left\{~i~ |~ c_{i} \neq 0,~ i=1, \ldots, n\right\}.
\]

For $\mathbf{c}_{1}, \mathbf{c}_{2} \in \mathcal{C}$, when $\operatorname{supp}(\mathbf{c}_{2}) \subseteq \operatorname{supp}(\mathbf{c}_{1})$, we say that $\mathbf{c}_{1}$ covers $\mathbf{c}_{2}$. A nonzero codeword $\mathbf{c} \in \mathcal{C}$ is minimal if it covers only the codeword $\lambda \mathbf{c}(\lambda \in \mathbb{F}_{q}^{*})$, but no other codewords in $\mathcal{C}$. The following lemma provides a sufficient condition for a linear code to be minimal.

\begin{lemma}\cite{A51}
	Every nonzero codeword of a linear code $\mathcal{C}$ over $\mathbb{F}_{2}$ is minimal provided that
	\[
	\frac{w_{\min }}{w_{\max }}>\frac{1}{2},
	\]
	where $w_{\min }$ and $w_{\max }$ are the minimal
	and maximal nonzero weights of the linear code $\mathcal{C}$ over $\mathbb{F}_{2}$, respectively.
\end{lemma}

\section{Some necessary Lemmas}

	In this section, we give some important auxiliary results which will be used in the sequel. Obviously, $(\operatorname{Tr}(a x^dy +bx))_{(x, y) \in D_{i}}$ is the zero codeword when $(a, b)=(0,0)$. Hence we will default $ab\neq0$ in the following.

\begin{lemma}
	For $a \in \mathbb{F}_{2^{m}}^{*}$, suppose that $f(x)=x^{2}+x+a \in \mathbb{F}_{2^{m}}[x]$ is a function from $\mathbb{F}_{2^{m}}$ to $\mathbb{F}_{2^m}$, then $f(x)=0$ has two nonzero solutions if and only if $a=r^2+r$, where $r \in \mathbb{F}_{2^m}\setminus \mathbb{F}_{2}$.
\end{lemma}
\begin{proof}
Let $x_{1}, x_{2} \in \mathbb{F}_{2^m}\setminus \mathbb{F}_{2}$ be the nonzero solutions of $f(x)=0$, and according to Vieta Theorem, $x_{1}+x_{2}=1, x_{1} x_{2}=a$, so that $a=x_{1}^2+x_{1}$. Conversely, if $a=r^2+r$, then obviously $r$ and $r+1$ are both two nonzero solutions of $f(x)=0$.
\end{proof}

By Lemma 3.1, we  have the following 
\begin{corollary}
	Let $\Upsilon_{1}=\left\{a=r^2+r,~ r\in \mathbb{F}_{2^m}^{*}\right\}$, then $\left|\Upsilon_{1}\right|=2^{m-1}-1$.
\end{corollary}

\begin{lemma}
	Let  $S_{1}=\displaystyle{\sum_{x \in \mathbb{F}_{2^m}^{*}}\sum_{y \in \mathbb{F}_{2^{m}}}}(-1)^{\operatorname{Tr}(a x^dy +b x)}$, then 
	
	$$S_{1}=\left\{\begin{array}{ll}
		-2^{m}, &   if  ~a=0,~b\neq0; \\
		0, &   if ~a\neq0,~b\in\mathbb{F}_{2^m}.
	\end{array}\right.$$
\end{lemma}
\begin{proof}
(i) If $a=0$ and $b\neq0$,~then
\begin{align}\label{C3}
	S_{1}=\sum_{x \in \mathbb{F}_{2^m}^{*}}(-1)^{\operatorname{Tr}(b x)} \sum_{y \in \mathbb{F}_{2^{m}}}(-1)^{\operatorname{Tr}(0)}=2^{m} \sum_{x \in \mathbb{F}_{2^m}^{*}}(-1)^{\operatorname{Tr}(bx)},
\end{align}
we know that 
\begin{align}\label{C4}
	\sum_{x \in \mathbb{F}_{2^m}^{*}}(-1)^{\operatorname{Tr}(bx)}=-1
\end{align} 
by (\ref{C2}), so that $S_{1}=-2^{m}$ by (\ref{C3})-(\ref{C4}).

(ii) If $a\neq0$ and $b\in\mathbb{F}_{2^m}$,~then
\begin{align}\label{C5}
	S_{1}=\sum_{x \in \mathbb{F}_{2^m}^{*}}(-1)^{\operatorname{Tr}(b x)} \sum_{y \in \mathbb{F}_{2^{m}}}(-1)^{\operatorname{Tr}(ax^d y)},
\end{align}
when $x\in\mathbb{F}_{2^m}^{*}$, we know that
\begin{align}\label{C6}
	\displaystyle{\sum_{y \in \mathbb{F}_{2^{m}}}}(-1)^{\operatorname{Tr}(ax^{d}y)}=0
\end{align}  
by (\ref{C1}), so that $S_{2}=0$ by (\ref{C5})-(\ref{C6}).
\end{proof}

\begin{lemma}
	Let  $S_{2}=\displaystyle{\sum_{x \in \mathbb{F}_{2^m}^{*}}\sum_{y \in \mathbb{F}_{2^{m}}^{*}}}(-1)^{\operatorname{Tr}(a x^dy +b x)}$, then 
	
	$$S_{2}=\left\{\begin{array}{ll}
		-2^{m}+1, &   if  ~a=0~ and~b\neq0~or~a\neq0~ and~b=0; \\
		1, &   if ~a\neq0,~b\neq0.
	\end{array}\right.$$
\end{lemma}
\begin{proof}
(i) If $a=0$ and $b\neq0$,~then
\begin{align}
	S_{2}=\sum_{x \in \mathbb{F}_{2^m}^{*}}(-1)^{\operatorname{Tr}(b x)} \sum_{y \in \mathbb{F}_{2^{m}}^{*}}(-1)^{\operatorname{Tr}(0)}=(2^{m}-1) \sum_{x \in \mathbb{F}_{2^m}^{*}}(-1)^{\operatorname{Tr}(bx)},\notag
\end{align}
in the same proof as that of (\ref{C4}), we can get  $S_{2}=-2^m+1.$

(ii) If $a\neq0$ and $b=0$,~then
\begin{align}
	S_{2}=\sum_{x \in \mathbb{F}_{2^m}^{*}}(-1)^{\operatorname{Tr}(0)} \sum_{y \in \mathbb{F}_{2^{m}}^{*}}(-1)^{\operatorname{Tr}(ax^d y)},\notag
\end{align}
in the same proof as that of (\ref{C4}), we can get  $S_{2}=-2^m+1.$	

(iii) If $a\neq0$ and $b\neq0$,~then
\begin{align}
	S_{2}=\sum_{x \in \mathbb{F}_{2^m}^{*}}(-1)^{\operatorname{Tr}(bx)} \sum_{y \in \mathbb{F}_{2^{m}}^{*}}(-1)^{\operatorname{Tr}(ax^d y)},\notag
\end{align}
in the same proof as that of (\ref{C4}), we can get  $S_{2}=1.$
\end{proof}

\begin{lemma}
	Let  $S_{3}=\displaystyle{\sum_{x \in \mathbb{F}_{2^m}^{*}}\sum_{y \in \mathbb{F}_{2^{m}}}}(-1)^{\operatorname{Tr}(x^{d+1}y+x^{d}y+x)+\operatorname{Tr}(ax^dy+b x)}$, then
	
	\footnotesize{$$S_{3}=\left\{\begin{array}{ll}
			0, &   if ~a=1,~b\in\mathbb{F}_{2^m};  \\
			2^{m}, &   if ~a=0,~b\neq0~and~\operatorname{Tr}(b+1)=0~or~ a\in\mathbb{F}_{2^m}\setminus \mathbb{F}_{2}, ~b\in\mathbb{F}_{2^m}~and~\operatorname{Tr}((a+1)(b+1))=0;\\
			-2^m,  &    if~a=0,~b\neq0~and ~\operatorname{Tr}(b+1)=1~or~a\in\mathbb{F}_{2^m}\setminus \mathbb{F}_{2}, ~b\in\mathbb{F}_{2^m}~and~\operatorname{Tr}((a+1)(b+1))=1.\\
		\end{array}\right.$$}
\end{lemma}
\begin{proof}
(i) If $a=0$ and $b\neq0$,~then
\begin{align}\label{C7}
	S_{3}&=\sum_{x \in \mathbb{F}_{2^m}^{*}}(-1)^{\operatorname{Tr}((b+1)x)} \sum_{y \in \mathbb{F}_{2^{m}}}(-1)^{\operatorname{Tr}(x^d(x+1)y)}\notag\\
	&=(-1)^{\operatorname{Tr}(b+1)} \sum_{y \in \mathbb{F}_{2^{m}}}(-1)^{\operatorname{Tr}(1(1+1)y)}\notag\\
	&~~~+\sum_{x \in \mathbb{F}_{2^m}\setminus\mathbb{F}_{2}}(-1)^{\operatorname{Tr}((b+1)x)} \sum_{y \in \mathbb{F}_{2^{m}}}(-1)^{\operatorname{Tr}(x^{d}(x+1)y)}\notag\\
	&=(-1)^{\operatorname{Tr}(b+1)}2^{m}=\left\{\begin{array}{ll}
		2^{m}, &  if~\operatorname{Tr}(b+1)=0; \\
		-2^{m}, &  if ~\operatorname{Tr}(b+1)=1.
	\end{array}\right.
\end{align}

(ii) If $a\neq0$ and $b\in\mathbb{F}_{2^m}$,~then
\begin{align}
	S_{3}&=\sum_{x \in \mathbb{F}_{2^m}^{*}} (-1)^{\operatorname{Tr}((b+1)x)}\sum_{y \in \mathbb{F}_{2^{m}}}(-1)^{\operatorname{Tr}(x^d(x+a+1)y)},\notag
\end{align}	
in the same proof as that of (\ref{C6})-(\ref{C7}), we can get 
\begin{align}
	S_{3}&=\left\{\begin{array}{ll}
		0, & if ~a=1;  \notag\\
		2^m, & if~ a\neq1, ~\operatorname{Tr}((a+1)(b+1))=0; \notag\\
		-2^m, & if~ a\neq1, ~\operatorname{Tr}((a+1)(b+1))=1. 
	\end{array}\right.	
\end{align}	
\end{proof}

\begin{lemma}
	Let  $S_{4}=\displaystyle{\sum_{x \in \mathbb{F}_{2^m}^{*}}\sum_{y \in \mathbb{F}_{2^{m}}^{*}}}(-1)^{\operatorname{Tr}( x^{d+2}y+x^{d+1}y)+\operatorname{Tr}(ax^dy+bx)}$, then
	
	\footnotesize{$$S_{4}=\left\{\begin{array}{ll}
			1, &   if  ~a\neq0,~a\notin\Upsilon_{1}~and~b\neq0~or~a\in\Upsilon_{1},~b\neq0~and~Tr(b)=1;\\
			2^{m}+1, &   if~a=0,~b\neq0~and~\operatorname{Tr}(b)=0~or~a\in\Upsilon_{1}~and~b=0;\\
			-2^m+1,  &   if~a=0,~b\neq0~and~\operatorname{Tr}(b)=1~or~a\neq0,~a\notin\Upsilon_{1}~and~b=0;\\
			2^{m+1}+1~or~-2^{m+1}+1,  &   if~a\in\Upsilon_{1},~b\neq0~and~\operatorname{Tr}(b)=0.\\
		\end{array}\right.$$}
\end{lemma}
\begin{proof}
(i) If $a=0$ and $b\neq0$,~then
\begin{align}
	S_{4}&=\sum_{x \in \mathbb{F}_{2^m}^{*}}(-1)^{\operatorname{Tr}(bx)} \sum_{y \in \mathbb{F}_{2^{m}}^{*}}(-1)^{\operatorname{Tr}(x^{d+1}(x+1)y)}\notag\\
	&=(-1)^{\operatorname{Tr}(b)} \sum_{y \in \mathbb{F}_{2^{m}}^{*}}(-1)^{\operatorname{Tr}(1(1+1)y)}\notag\\
	&+\sum_{x \in \mathbb{F}_{2^m}\setminus\mathbb{F}_{2}}(-1)^{\operatorname{Tr} (bx)} \sum_{y \in \mathbb{F}_{2^{m}}^{*}}(-1)^{\operatorname{Tr}(x^{d+1}(x+1)y)},\notag
\end{align} 	
in the same proof as that of (\ref{C4}) and (\ref{C7}), we can get 
\begin{align}
	S_{4}&=\left\{\begin{array}{ll}
		2^{m}+1, &  if ~ \operatorname{Tr}(b)=0; \notag\\
		-2^{m}+1, &  if ~ \operatorname{Tr}(b)=1.\notag
	\end{array}\right.
\end{align}

(ii) If $a\neq0$ and $b=0$,~then
\begin{align}
	S_{4}&=\sum_{x \in \mathbb{F}_{2^m}^{*}}(-1)^{\operatorname{Tr}(0)} \sum_{y \in \mathbb{F}_{2^{m}}^{*}}(-1)^{\operatorname{Tr}(x^{d}(x^2+x+a)y)},\notag
\end{align}
when $a\notin\Upsilon_{1}$, in the same proof as that of (\ref{C4}), we can get $S_{4}=-2^{m}+1$;
when $a\in\Upsilon_{1}$, we can get 
\begin{align}
	S_{4}&=((-1)^{\operatorname{Tr}(0)}+ (-1)^{\operatorname{Tr}(0)})(2^m-1)+(-1)\sum_{x \in \mathbb{F}_{2^m}^{*}\setminus\left\{x_{1},x_{2}\right\}}(-1)^{\operatorname{Tr}(0)} \notag\\
	&=2(2^m-1)-(2^m-3)\notag\\
	&=2^{m}+1,\notag
\end{align}
where $x_{1}$ and $x_{2}$ are both nonzero solutions of $x^{2}+ x+a=0$. Thus we have
\begin{align}
	S_{4}&=\left\{\begin{array}{ll}
		2^{m}+1, &  if ~a\in\Upsilon_{1};\notag \\
		-2^{m}+1, &  if  ~a\notin\Upsilon_{1}.\notag
	\end{array}\right.
\end{align}

(iii) If $a\neq0$ and $b\neq0$,~then
\begin{align}
	S_{4}&=\sum_{x \in \mathbb{F}_{2^m}^{*}}(-1)^{\operatorname{Tr}(bx)} \sum_{y \in \mathbb{F}_{2^{m}}^{*}}(-1)^{\operatorname{Tr}(x^{d}(x^2+x+a)y)},\notag
\end{align}
when $a\notin\Upsilon_{1}$, in the same proof as that of (\ref{C4}), we can get $S_{4}=1$;
when $a\in\Upsilon_{1}$, we can get 
\begin{align}
	S_{4}&=((-1)^{\operatorname{Tr}(bx_{1})}+ (-1)^{\operatorname{Tr}(bx_{2})})(2^m-1)+(-1)\sum_{x \in \mathbb{F}_{2^m}^{*}\setminus\left\{x_{1},x_{2}\right\}}(-1)^{\operatorname{Tr}(bx)} \notag\\
	&=((-1)^{\operatorname{Tr}(bx_{1})}+(-1)^{\operatorname{Tr}(bx_{2})})(2^m-1)-(-1-(-1)^{\operatorname{Tr}(bx_{1})}-(-1)^{\operatorname{Tr}(bx_{2})}),\notag	
\end{align}
according to Vieta Theorem, $x_{1}+x_{2}=1$. And then we can derive
\begin{align}
	(-1)^{\operatorname{Tr}(bx_{1})}+(-1)^{\operatorname{Tr}(bx_{2})}&=\left\{\begin{array}{ll}
		2~ or~-2, & if ~(-1)^{\operatorname{Tr}(b)}=1;  \notag\\
		0, & if ~(-1)^{\operatorname{Tr}(b)}=-1.
	\end{array}\right.\notag\\
	&=\left\{\begin{array}{ll}
		2 ~or~-2, & if ~\operatorname{Tr}(b)=0;  \notag\\
		0, & if ~\operatorname{Tr}(b)=1.
	\end{array}\right. 
\end{align}	

\noindent Thus we have
$$	S_{4}=\left\{\begin{array}{ll}
	2^{m+1}+1~ or ~-2^{m+1}+1, & if~a\in\Upsilon_{1},~\operatorname{Tr}(b)=0;  \notag\\
	1, & if ~a\notin\Upsilon_{1} ~or~ a\in\Upsilon_{1} ~and~\operatorname{Tr}(b)=1.
\end{array}\right.$$
\end{proof}

\begin{lemma}	Let  $S_{5}=\displaystyle{\sum_{x \in \mathbb{F}_{2^m}^{*}}\sum_{y \in \mathbb{F}_{2^{m}}}}(-1)^{\operatorname{Tr}(x^{d+1}y)+\operatorname{Tr}(ax^dy+b x)},$
	\begin{align}
		&S_{6}=\displaystyle{\sum_{x \in \mathbb{F}_{2^m}^{*}}\sum_{y \in \mathbb{F}_{2^{m}}}}(-1)^{\operatorname{Tr}(x^{d}y)-1+\operatorname{Tr}(ax^dy+b x)},\notag\\ 
		&S_{7}=\displaystyle{\sum_{x \in \mathbb{F}_{2^m}^{*}}\sum_{y \in \mathbb{F}_{2^{m}}}}(-1)^{\operatorname{Tr}(x^{d+1}y)+\operatorname{Tr}(x^{d}y)-1+\operatorname{Tr}(ax^dy+b x)}.\notag
	\end{align}
	Then
	\begin{align}
		&S_{5}=\left\{\begin{array}{ll}
			0, &   if ~a=0,~b\neq0;  \\
			2^{m}, &   if ~a\neq0,~b\in\mathbb{F}_{2^m}~and~\operatorname{Tr}(ab)=0;\\
			-2^m,  &   if ~a\neq0,~b\in\mathbb{F}_{2^m}~and~\operatorname{Tr}(ab)=1.\\
		\end{array}\right.\notag\\
		&S_{6}=\left\{\begin{array}{ll}
			0, &   if ~a\neq1,~b\in\mathbb{F}_{2^m};  \\
			2^{m}, &   if ~a=1,~b\neq0;\\
			-2^{2m}+2^m,  &   if ~a=1,~b=0.\\
		\end{array}\right.\notag\\
		&S_{7}=\left\{\begin{array}{ll}
			0, &   if ~a=1,~b\in\mathbb{F}_{2^m};  \\
			2^{m}, &   if ~a\neq1,~b\in\mathbb{F}_{2^m}~and~\operatorname{Tr}((a+1)b)=1;\\
			-2^m,  &   if ~a\neq1,~b\in\mathbb{F}_{2^m}~and~\operatorname{Tr}((a+1)b)=0.\\
		\end{array}\right.\notag
	\end{align}
	
\end{lemma}
\begin{proof}
(i) If $a=0$ and $b\neq0$,~then
\begin{align}
	S_{5}=\sum_{x \in \mathbb{F}_{2^m}^{*}}(-1)^{\operatorname{Tr}(bx)} \sum_{y \in \mathbb{F}_{2^{m}}}(-1)^{\operatorname{Tr}(x^{d+1}y)},\notag
\end{align}
in the same proof as that of (\ref{C6}), we have  $S_{5}=0$. Next we compute
\begin{align}
	S_{6}=(-1)^{-1}\sum_{x \in \mathbb{F}_{2^m}^{*}}(-1)^{\operatorname{Tr}(bx)} \sum_{y \in \mathbb{F}_{2^{m}}}(-1)^{\operatorname{Tr}(x^{d}y)},\notag
\end{align}
in the same proof as that of (\ref{C6}), we have  $S_{6}=0$. Now we consider

\begin{align}
	S_{7}&=(-1)^{-1}\sum_{x \in \mathbb{F}_{2^m}^{*}}(-1)^{\operatorname{Tr}(bx)} \sum_{y \in \mathbb{F}_{2^{m}}}(-1)^{\operatorname{Tr}(x^{d}(x+1)y)},\notag
\end{align}
in the same proof as that of (\ref{C7}), we can get 
\begin{align}
	S_{7}&=\left\{\begin{array}{ll}
		2^m, & if ~\operatorname{Tr}(b)=1; \notag\\
		-2^m, & if ~\operatorname{Tr}(b)=0. 
	\end{array}\right.
\end{align}

(ii) If $a\neq0$ and $b\in\mathbb{F}_{2^m}$,~then
\begin{align}
	S_{5}&=\sum_{x \in \mathbb{F}_{2^m}^{*}} (-1)^{\operatorname{Tr}(bx)}\sum_{y \in \mathbb{F}_{2^{m}}}(-1)^{\operatorname{Tr}(x^d(x+a)y)},\notag
\end{align}	
in the same proof as that of (\ref{C7}), we can get 
\begin{align}
	S_{5}&=\left\{\begin{array}{ll}
		2^m, & if ~\operatorname{Tr}(ab)=0; \notag\\
		-2^m, & if ~\operatorname{Tr}(ab)=1. 
	\end{array}\right.
\end{align}	
Next we compute
\begin{align}
	S_{6}&=(-1)^{-1}\sum_{x \in \mathbb{F}_{2^m}^{*}} (-1)^{\operatorname{Tr}(bx)}\sum_{y \in \mathbb{F}_{2^{m}}}(-1)^{\operatorname{Tr}(x^d(a+1)y)}.\notag
\end{align}	
When $a\neq1$, in the same proof as that of  (\ref{C6}) we can get  $S_{6}=0$;
when $a=1$ and $b\neq0$, in the same proof as that of  (\ref{C3})-(\ref{C4}) we can get  $S_{6}=2^m$;
when $a=1$ and $b=0$, we can get 
\begin{align}\label{C8}
	S_{6}&=(-1)^{-1}\sum_{x \in \mathbb{F}_{2^m}^{*}} (-1)^{\operatorname{Tr}(0)}\sum_{y \in \mathbb{F}_{2^{m}}}(-1)^{\operatorname{Tr}(0)}=-(2^m-1)2^m=-2^{2m}+2^m.\notag
\end{align}	

\noindent Thus we can get
\begin{align}
	S_{6}&=\left\{\begin{array}{ll}
		-2^{2m}+2^m, & if ~a=1,~b=0; \notag\\
		2^m, & if ~a=1,~b\neq0; \notag\\
		0, & if ~a\neq1. 
	\end{array}\right.
\end{align}	
Now we consider
\begin{align}
	S_{7}&=(-1)^{-1}\sum_{x \in \mathbb{F}_{2^m}^{*}} (-1)^{\operatorname{Tr}(bx)}\sum_{y \in \mathbb{F}_{2^{m}}}(-1)^{\operatorname{Tr}(x^d(x+a+1)y)},\notag
\end{align}
in the same proof as that of (\ref{C6})-(\ref{C7}), we can get 
\begin{align}
	S_{7}&=\left\{\begin{array}{ll}
		2^m, & if ~a\neq1,~\operatorname{Tr}((a+1)b)=1; \notag\\
		-2^m, & if ~a\neq1,~\operatorname{Tr}((a+1)b)=0; \notag\\
		0, & if ~a=1. 
	\end{array}\right.
\end{align}
\end{proof}

  Lemma $3.8$ is useful for determining whether  $C_{D_{i}}$ is projective, where $i=1,2,3$.
\begin{lemma}
	For $D_{i} \subseteq \mathbb{F}_{p^{m}} \times \mathbb{F}_{p^{m}} \backslash\{0,0\}$, $C_{D_{i}}$ is projective, where $i=1,2,3$.
\end{lemma}
\begin{proof}
	Suppose that the Hamming distance of $C_{D_{i}}^{\perp}$ is $d^{\perp}$, where $i=1,2,3$. By $(0,0) \notin D_{i}$, we  get $d^{\perp}>1$ directly.
	If $d^{\perp}=2$, then there exists some $\left(x_{t}, y_{t}\right) \in D_{i}(t=1,2)$ with $\left(x_{1}, y_{1}\right) \neq \left(x_{2}, y_{2}\right)$ such that
	\[
	\operatorname{Tr}(a x_{1}^{d} y_{1}+b x_{1})+\operatorname{Tr}(a x_{2}^{d} y_{2}+b x_{2})=0 \text { for any }(a, b) \in \mathbb{F}_{p^{m}} \times \mathbb{F}_{p^{m}} \text {, }
	\]
	namely,
	\[
	\operatorname{Tr}(a(x_{1}^{d} y_{1}+x_{2}^{d} y_{2})+b(x_{1}+x_{2}))=0 \text { for any }(a, b) \in \mathbb{F}_{p^{m}} \times \mathbb{F}_{p^{m}}\text {, }
	\]
	which leads to $\left(x_{1}, y_{1}\right)=\left(x_{2}, y_{2}\right)$, thus $d^{\perp} \neq 2$.
\end{proof}

\section{Our main results and their proofs}\label{h3}
In this section, we give our main results, their proofs and some examples.
\subsection{Our main results}
\begin{theorem}
	For $m$ odd and $m\geq3$, suppose that ${\mathcal{C}_{D_{1}}}$ and ${D_{1}}$ are defined by $(1)$ and $(2)$, respectively, then ${\mathcal{C}_{D_{1}}}$ is a ${[~2^m(2^{m-1}-1), ~2 m,~ 2^{m-2}(2^{m}-3)~]}$ projective four-weight minimal linear code with weight distribution given in Table $1$. 
\end{theorem}
$$\begin{array}{l}
	\footnotesize{\text { Table 1. The weight distribution of } \mathcal{C}_{D_{1}}}\\
	\begin{array}{cc}
		\hline \text { weight } & \text { frequency } \\
		\hline 0~ & ~1 \\
		2^{m-2}(2^{m}-3)~ & ~2^{m}(2^{m-1}-1) \\
		2^{m-1}(2^{m-1}-1)~ & ~3\cdot2^{m-1} \\
		2^{m-2}(2^{m}-1)~ & ~2^{m}(2^{m-1}-1) \\
		2^{2m-2}~ & ~2^{m-1}-1 \\
		\hline
	\end{array}
\end{array}$$

\begin{theorem}
		For $m\geq4$, suppose that ${\mathcal{C}_{D_{2}}}$ and ${D_{2}}$ are defined by $(1)$ and $(3)$, respectively, then ${\mathcal{C}_{D_{2}}}$ is a $[~2^{2m-1}-2^m+1, ~2 m,~ 2^{m}(2^{m-2}-1)~]$  projective tree-weight minimal linear code with weight distribution given in Table $2$.
\end{theorem}
$$\begin{array}{l}
	~~~\footnotesize{\text { Table 2. The weight distribution of } \mathcal{C}_{D_{2}}}\\
	\begin{array}{cc}
		\hline \text { weight } & \text { frequency } \\
		\hline 0~ & ~1 \\
		2^{m}(2^{m-2}-1)~ & ~2^{m-2}(2^{m-1}-3)+1 \\
		2^{m-1}(2^{m-1}-1)~ & ~3\cdot2^{2m-2}-2 \\
		2^{2m-2}~ & ~2^{m-2}(2^{m-1}+3) \\
		\hline
	\end{array}
\end{array}$$

\begin{theorem}
	Suppose that ${\mathcal{C}_{D_{3}}}$ and ${D_{3}}$ are defined by $(1)$ and $(4)$, respectively, then ${\mathcal{C}_{D_{3}}}$ is a $[~2^{m-1}(2^{m-1}-1), ~2 m,~ 2^{m-1}(2^{m-2}-1)~]$  projective four-weight linear code with weight distribution given in Table $3$.
\end{theorem}
$$\begin{array}{l}
	\footnotesize{\text { Table 3. The weight distribution of } \mathcal{C}_{D_{3}}}\\
	\begin{array}{cc}
		\hline \text { weight } & \text { frequency } \\
		\hline 0~ & ~1 \\
		2^{m-1}(2^{m-2}-1)~ & ~2^{2m-2}-1 \\
		2^{ m-2}(2^{m-1}-1)~ & ~2^{2m-1} \\
		2^{2m-3}~ & ~2^{2m-2}-1 \\
		2^{m-1}(2^{m-1}-1)~ & ~1 \\
		\hline
	\end{array}
\end{array}$$\\

\subsection{The proofs of Theorems 4.1-4.3}

\noindent{\bfseries The proof for Theorem $4.1$.}

According to the definition, the length of  $\mathcal{C}_{D_{1}}$ equals to 
$$\begin{aligned}
	n_{1} & =\left|D_{1}\right| \\
	& =\frac{1}{2} \sum_{z \in \mathbb{F}_{2}} \sum_{x \in \mathbb{F}_{2^{m}}^{*}} \sum_{y \in \mathbb{F}_{2^m}}(-1)^{z \operatorname{Tr}(x^{d+1}y+x^dy+x)} \\
	& =\left(2^{m}-1\right) 2^{m-1}+\frac{1}{2} \sum_{x \in \mathbb{F}_{2^{m}}^{*}}(-1)^{\operatorname{Tr}(x)} \sum_{y \in \mathbb{F}_{2^{m}} }(-1)^{\operatorname{Tr}(x^d(x+1)y)} \\
	& =(2^{m}-1) 2^{m-1}-2^{m-1} \\
	& =2^m(2^{m-1}-1).
\end{aligned}$$

Let
$$
N_{1}=\#\left\{(x, y) \in \mathbb{F}_{2^{m}}^{*} \times \mathbb{F}_{2^{m}}: \operatorname{Tr}(x^{d+1}y+x^dy+x)=0 \text { and } \operatorname{Tr}(ax^dy+bx)=0\right\},
$$
then the weight of the nonzero codeword ${\mathbf{c}}$ of the linear code ${\mathcal{C}_{D_{1}}}$ is
$$\begin{aligned}
	w t_{1}(\mathbf{c}) & =n_{1}-N_{1} \\
	& =n_{1}-\frac{1}{2^{2}} \sum_{z_{1} \in \mathbb{F}_{2}} \sum_{z_{2} \in \mathbb{F}_{2}} \sum_{x \in \mathbb{F}_{2^m}^{*} } \sum_{y \in \mathbb{F}_{2^{m}}}(-1)^{z_{1} \operatorname{Tr}(x^{d+1}y+x^dy+x)+z_{2} \operatorname{Tr}(ax^dy+bx)} \\
	& =n_{1}-\frac{1}{2} n_{1}-\frac{1}{2^{2}} \sum_{x \in \mathbb{F}_{2^m}^{*}} \sum_{y \in \mathbb{F}_{2^{m}}}(-1)^{\operatorname{Tr}(ax^dy+bx)}\\
	&\quad-\frac{1}{2^{2}} \sum_{x \in \mathbb{F}_{2^{m}}^{*}} \sum_{y \in \mathbb{F}_{2^{m}}}(-1)^{\operatorname{Tr}( x^{d+1}y+x^dy+x)+\operatorname{Tr}(ax^dy+bx)} \\
	& =2^{m-1}(2^{m-1}-1)-\frac{1}{2^{2}} S_{1}-\frac{1}{2^{2}} S_{3}.
\end{aligned}$$
Plugging in Lemma $3.3$ and Lemma $3.5$, we can get 
$$\footnotesize {w t_{1}(\mathbf{c})=\left\{\begin{array}{ll}
	2^{m-2}(2^m-3), &  if~a\in\mathbb{F}_{2^m}\setminus \mathbb{F}_{2}, ~b\in\mathbb{F}_{2^m} ~and~\operatorname{Tr}((a+1)(b+1))=0;\\
	2^{m-1}(2^{m-1}-1), &  if~a=1~and~b\in\mathbb{F}_{2^m}~or~ a=0,~b\neq0~and ~\operatorname{Tr}(b+1)=0;\\
	2^{m-2}(2^m-1), &  if~a\in\mathbb{F}_{2^m}\setminus \mathbb{F}_{2}, ~b\in\mathbb{F}_{2^m} ~and~\operatorname{Tr}((a+1)(b+1))=1; \\
	2^{2m-2}, &  if~a=0,~ b\neq0~and~\operatorname{Tr}(b+1)=1.
\end{array}\right.}$$

Now we give the weight $w_{1}=2^{m-2}(2^m-3),~w_{2}=2^{m-1}(2^{m-1}-1),~w_{3}=2^{2m-2}$ and $w_{4}=2^{2m-2}$, their corresponding frequencies are $A_{w_{1}},~A_{w_{2}},~A_{w_{3}}$ and ${A_{w_{4}}}$, respectively. Hence we can obtain
\begin{align}	
	A_{w_{1}}=2^m(2^{m-1}-1), ~ A_{w_{2}}=3\cdot2^{m-1}, 
	~A_{w_{3}}=2^m(2^{m-1}-1), ~ A_{w_{4}}=2^{m-1}-1. \notag
\end{align}

According  to  $D_{1}\subseteq \mathbb{F}_{2^{m}}^{*} \times \mathbb{F}_{2^{m}}$ and Lemma $3.8$, we can get that $\mathcal{C}_{D_{1}}$ is a projective code.

 When $m\geq3$, then for $\mathcal{C}_{D_{1}}$, it holds that
\[
\frac{w_{\min }}{w_{\max }}=\frac{2^{m-2}(2^m-3)}{2^{2m-2}}=1-\frac{3}{2^{m}}>\frac{1}{2},
\]
now by Lemma $2.3$, all nonzero codewords in $\mathcal{C}_{D_{1}}$ are minimal, and so their dual codes can be employed to construct secret sharing schemes with interesting access structures.\\

\noindent{\bfseries The proof for Theorem $4.2$.}

According to the definition, the length of  $\mathcal{C}_{D_{2}}$ equals to 
$$\begin{aligned}
	n_{2} & =\left|D_{2}\right| \\
	& =\frac{1}{2} \sum_{z \in \mathbb{F}_{2}} \sum_{x \in \mathbb{F}_{2^{m}}^{*}} \sum_{y \in \mathbb{F}_{2^m}^{*}}(-1)^{z \operatorname{Tr}(x^{d+2}y+x^{d+1}y)} \\
	& =\frac{1}{2} (2^{m}-1) (2^{m}-1)+\frac{1}{2} \sum_{x \in \mathbb{F}_{2^{m}}^{*}}\sum_{y \in \mathbb{F}_{2^m}^{*}}(-1)^{\operatorname{Tr}(x^{d+1}(x+1)y)} \\
	&=\frac{1}{2}(2^{m}-1) (2^{m}-1)+\frac{1}{2}(2^{m}-1)-\frac{1}{2}(2^{m}-2) \\
	& =2^{2m-1}-2^m+1.
\end{aligned}$$

Let
$$
N_{2}=\#\left\{(x, y) \in \mathbb{F}_{2^m}^{*} \times \mathbb{F}_{2^{m}}^{*}: \operatorname{Tr}(x^{d+2}y+x^{d+1}y)=0 \text { and } \operatorname{Tr}(ax^dy+bx)=0\right\},
$$
then the weight of the nonzero codeword ${\mathbf{c}}$ of the linear code ${\mathcal{C}_{D_{2}}}$ is
$$\begin{aligned}
	w t_{2}(\mathbf{c}) & =n_{2}-N_{2} \\
	& =n_{2}-\frac{1}{2^{2}} \sum_{z_{1} \in \mathbb{F}_{2}} \sum_{z_{2} \in \mathbb{F}_{2}} \sum_{x \in \mathbb{F}_{2^m}^{*} } \sum_{y \in \mathbb{F}_{2^{m}}^{*}}(-1)^{z_{1} \operatorname{Tr}(x^{d+2}y+x^{d+1}y)+z_{2} \operatorname{Tr}(ax^dy+bx)} \\
	& =n_{2}-\frac{1}{2} n_{2}-\frac{1}{2^{2}} \sum_{x \in \mathbb{F}_{2^m}^{*}} \sum_{y \in \mathbb{F}_{2^{m}}^{*}}(-1)^{\operatorname{Tr}(ax^dy+bx)}\\
	&\quad-\frac{1}{2^{2}} \sum_{x \in \mathbb{F}_{2^{m}}^{*}} \sum_{y \in \mathbb{F}_{2^{m}}^{*}}(-1)^{\operatorname{Tr}( x^{d+2}y+x^{d+1}y)+\operatorname{Tr}(ax^dy+bx)} \\
	& =2^{m-1}(2^{m-1}-1)+\frac{1}{2}-\frac{1}{2^{2}} S_{2}-\frac{1}{2^{2}} S_{4}.
\end{aligned}$$
Plugging in Lemma $3.4$ and Lemma $3.6$, we can get 
$$w t_{2}(\mathbf{c})=\left\{\begin{array}{ll}
	2^{m-1}(2^{m-1}-1), &  if~a\neq0,~a\notin\Upsilon_{1}~and~b\neq0,\\ &or~a\in\Upsilon_{1},~b\neq0~and~\operatorname{Tr}(b)=1,\\
	&or~a=0,~b\neq0~and~\operatorname{Tr}(b)=0,\\
	&or~a\in\Upsilon_{1},~b=0;\\
	2^{m}(2^{m-2}-1)~or~2^{2m-2}, &  otherwise.
\end{array}\right.$$

Now we give the weight $w_{5}=2^{m}(2^{m-2}-1),~w_{6}=2^{m-1}(2^{m-1}-1)$ and $w_{7}=2^{2m-2}$, their corresponding frequencies are $A_{w_{5}},~A_{w_{6}}$ and ${A_{w_{7}}}$, respectively.  According to Lemma $2.2$, Lemma $3.1$ and Corollary $3.2$, we can obtain
\begin{align}	
	A_{w_{5}}=2^{m-2}(2^{m-1}-3)+1,~A_{w_{6}}=3\cdot2^{2m-2}-2,~A_{w_{7}}=2^{m-2}(2^{m-1}+3).\notag
\end{align}

According  to  $D_{2}\subseteq \mathbb{F}_{2^{m}}^{*} \times \mathbb{F}_{2^{m}}^{*}$ and Lemma $3.8$, we can get that $\mathcal{C}_{D_{2}}$ is a projective code.

When $m\geq4$, then for $\mathcal{C}_{D_{2}}$, it holds that
\[
\frac{w_{\min }}{w_{\max }}=\frac{2^{m}(2^{m-2}-1)}{2^{2m-2}}=1-\frac{1}{2^{m-2}}>\frac{1}{2},
\]
now by Lemma $2.3$, all nonzero codewords in $\mathcal{C}_{D_{2}}$ are minimal, and so their dual codes can be employed to construct secret sharing schemes with interesting access structures.\\

\noindent{\bfseries The proof for Theorem $4.3$.}

According to the definition, the length of  $\mathcal{C}_{D_{3}}$ equals to 
$$\begin{aligned}
	n_{3} & =\left|D_{3}\right| \\
	& =\frac{1}{2^2} \sum_{z_{1} \in \mathbb{F}_{2}} \sum_{z_{2} \in \mathbb{F}_{2}} \sum_{x \in \mathbb{F}_{2^{m}}^{*} } \sum_{y \in \mathbb{F}_{2^{m}}}(-1)^{z_{1} \operatorname{Tr}(x^{d+1}y)+z_{2} (\operatorname{Tr}(x^dy)-1)} \\
	& =\left(2^{m}-1\right) 2^{m-2}+\frac{1}{2^2}\sum_{x \in \mathbb{F}_{2^{m}}^{*} } \sum_{y \in \mathbb{F}_{2^{m}}}(-1)^{ \operatorname{Tr}(x^{d+1}y)} \\
	&~+\frac{1}{2^2}\sum_{x \in \mathbb{F}_{2^{m}}^{*} } \sum_{y \in \mathbb{F}_{2^{m}}}(-1)^{\operatorname{Tr}(x^dy)-1}+\frac{1}{2^2} \sum_{x \in \mathbb{F}_{2^{m}}^{*}}\sum_{y \in \mathbb{F}_{2^m} }(-1)^{\operatorname{Tr}(x^d(x+1)y)-1}\\
	& =(2^{m}-1) 2^{m-2}+0+0-2^{m-2} \\
	& =2^{m-1}(2^{m-1}-1).
\end{aligned}$$

Let
$$
N_{3}=\#\left\{(x, y) \in \mathbb{F}_{2^{m}}^{*} \times \mathbb{F}_{2^{m}}: \operatorname{Tr}(x^{d+1}y)=0,~\operatorname{Tr}(x^{d}y)=1\text { and } \operatorname{Tr}(ax^dy+bx)=0\right\},
$$
then the weight of the nonzero codeword ${\mathbf{c}}$ of the linear code ${\mathcal{C}_{D_{3}}}$ is
$$\begin{aligned}
	w t_{3}(\mathbf{c}) & =n_{3}-N_{3} \\
	& =n_{3}-\frac{1}{2^{3}} \sum_{z_{1} \in \mathbb{F}_{2}} \sum_{z_{2} \in \mathbb{F}_{2}} \sum_{z_{3} \in \mathbb{F}_{2}}\sum_{x \in \mathbb{F}_{2^m}^{*}} \sum_{y \in \mathbb{F}_{2^{m}}}(-1)^{z_{1} \operatorname{Tr}(ax^dy+bx)+z_{2} \operatorname{Tr}(x^{d+1}y)+z_{3} (\operatorname{Tr}(x^dy)-1)} \\
	& =n_{3}-\frac{1}{2} n_{3}-\frac{1}{2^{3}} \sum_{x \in \mathbb{F}_{2^m}^{*}} \sum_{y \in \mathbb{F}_{2^{m}}}(-1)^{\operatorname{Tr}(ax^dy+bx)}-\frac{1}{2^{3}} \sum_{x \in \mathbb{F}_{2^{m}}^{*}} \sum_{y \in \mathbb{F}_{2^{m}}}(-1)^{\operatorname{Tr}( x^{d+1}y)+\operatorname{Tr}(ax^dy+bx)}\\
	&\quad-\frac{1}{2^{3}} \sum_{x \in \mathbb{F}_{2^{m}}^{*}} \sum_{y \in \mathbb{F}_{2^{m}}}(-1)^{\operatorname{Tr}( x^{d}y)-1+\operatorname{Tr}(ax^dy+bx)}\\
	&\quad-\frac{1}{2^{3}} \sum_{x \in \mathbb{F}_{2^{m}}^{*}} \sum_{y \in \mathbb{F}_{2^{m}}}(-1)^{\operatorname{Tr}( x^{d+1}y)+\operatorname{Tr}( x^{d}y)-1+\operatorname{Tr}(ax^dy+bx)}  \\
	& =2^{m-2}(2^{m-1}-1)-\frac{1}{2^{3}} S_{1}-\frac{1}{2^{3}} S_{5}-\frac{1}{2^{3}} S_{6}-\frac{1}{2^{3}} S_{7}.
\end{aligned}$$
Plugging in Lemma $3.3$ and Lemma $3.7$, we can get 
$$\footnotesize {w t_{3}(\mathbf{c})=\left\{\begin{array}{ll}
		2^{m-1}(2^{m-2}-1), &  if~a=1,~b\neq0~and~\operatorname{Tr}(ab)=0,\\
		&or~a\in\mathbb{F}_{2^m}\setminus \mathbb{F}_{2},~b\in\mathbb{F}_{2^m},~\operatorname{Tr}(ab)=0~and~\operatorname{Tr}((a+1)b)=1;\\
		2^{m-2}(2^{m-1}-1), &  if~a=1,~b\neq0~and~\operatorname{Tr}(ab)=1,\\
		&or~ a=0,~b\neq0~and~\operatorname{Tr}((a+1)b)=1,\\
		&or~a\in\mathbb{F}_{2^m}\setminus \mathbb{F}_{2},~b\in\mathbb{F}_{2^m},~\operatorname{Tr}(ab)=1~and~\operatorname{Tr}((a+1)b)=1,\\
		&or~a\in\mathbb{F}_{2^m}\setminus \mathbb{F}_{2},~b\in\mathbb{F}_{2^m},~\operatorname{Tr}(ab)=0~and~\operatorname{Tr}((a+1)b)=0;\\
		2^{2m-3}, & if~ a=0,~b\neq0~and~\operatorname{Tr}((a+1)b)=0, \\
		&or~a\in\mathbb{F}_{2^m}\setminus \mathbb{F}_{2},~b\in\mathbb{F}_{2^m},~\operatorname{Tr}(ab)=1~and~\operatorname{Tr}((a+1)b)=0;\\
		2^{m-1}(2^{m-1}-1), &if~ a=1, ~b=0~and~\operatorname{Tr}(ab)=0;\\
		2^{m-2}(2^{m}-1), &if~ a=1, ~b=0~and~\operatorname{Tr}(ab)=1.
	\end{array}\right.}$$

Now we give the weight $w_{8}=2^{m-2}(2^m-3),~ w_{9}=2^{m-1}(2^{m-1}-1),~w_{10}=2^{2m-2}$, $w_{11}=2^{2m-2}$ and $w_{12}=2^{m-2}(2^{m}-1)$, their corresponding frequencies are $A_{w_{8}}, ~A_{w_{9}}, ~ A_{w_{10}}$,~ ${A_{w_{11}}}$ and ${A_{w_{12}}}$, respectively. Hence we can obtain
\begin{align}	
	A_{w_{8}}=2^{2m-2}-1,~A_{w_{9}}=2^{2m-1},~A_{w_{10}}=2^{2m-2}-1,	~A_{w_{11}}=1,~A_{w_{12}}=0.\notag
\end{align}

According  to  $D_{3}\subseteq \mathbb{F}_{2^{m}}^{*} \times \mathbb{F}_{2^{m}}$ and Lemma $3.8$, we can get that $\mathcal{C}_{D_{3}}$ is a projective code.\\

\subsection{Some examples}

$\quad\quad\quad\quad\quad\quad\quad\quad\quad\quad\quad\text { Table } 4.~ \text { Some examples }$ 
\begin{table}[!ht]
	\centering 
		\begin{tabular}{|c|c|c|c|}
		\hline
		$m$& Theorem& parameters &weight enumerator \\ \hline 
		\multirow{3}{*}{$m=3$}&Theorem 4.1&[24,6,10]&$ 1+24z^{10}+12z^{12}+24z^{14}+3z^{16}$\\ 
		\cline{2-4} 
		&Theorem 4.2&[25,6,8] & $1+3z^8+46z^{12}+14z^{16}$\\
		\cline{2-4} 
		&Theorem 4.3& [12,6,4]& $1+15z^{4}+32z^{6}+15z^{8}+z^{12}$ \\
		\hline
		\multirow{2}{*}{$m=4$}&Theorem 4.2&[113,8,48] & $1+21z^{48}+190z^{56}+44z^{64}$\\ 
		\cline{2-4} 
		&Theorem 4.3&[56,8,24] & $1+63z^{24}+128z^{28}+ 63z^{32}+z^{56}$ \\
		\hline 
		\multirow{3}{*}{$m=5$}&Theorem 4.1& [480,10,232]& $1+480z^{232}+48z^{240}+480 z^{248}+15z^{256}$ \\ 
		\cline{2-4} 
		&Theorem 4.2&[481,10,224] & $1+105z^{224}+766z^{240}+152z^{256}$\\
		\cline{2-4} 
		&Theorem 4.3&[240,10,112] & $1+255z^{112}+512z^{120}+ 255z^{128}+z^{240}$\\
		\hline  
	\end{tabular}
\end{table}

\begin{remark}
	According to the Griesmer bound \cite{A6}, the code  $[12, 6, 4]$ is optimal.
\end{remark}

\section{Conclusions}
In this manuscript, a class of projective three-weight linear codes and two classes of projective four-weight linear codes over $\mathbb{F}_{2}$ are constructed. The projective three-weight linear code and one class of projective four-weight linear codes (Theorem $4.1$) are suitable for applications in secret sharing schemes. In particular, according to the Griesmer bound we obtain a class of projective four-weight linear codes  (Theorem $4.3$) which are optimal  when $m=3$.


\begin{thebibliography}{99}

\bibitem{A4}A. Calderbank, J. Goethals, Three-weight codes and association schemes, Philips J. Res. 39 (4–5) (1984) 143–152.
\bibitem{A51}A. Ashikhmin, A. Barg, Minimum vectors in linear codes, IEEE Trans. Inf. Theory 44 (1998), 2010-2017.
\bibitem{A2}C. Carlet, C. Ding, J. Yuan, Linear codes from perfect nonlinear mappings and their secret sharing schemes, IEEE Trans. Inf. Theory 51 (6) (2005) 2089–2102.
\bibitem{A5}C. Ding, X. Wang, A coding theory construction of new systematic authentication codes, Theor. Comput. Sci. 330 (1) (2005) 81–99.
\bibitem{A13}C. Ding, Linear codes from some 2-designs, IEEE Trans. Inf. Theory 61 (6) (2015) 3265–3275.
\bibitem{A26}C. Li, S. Bae, S. Yang, Some two-weight and three-weight linear codes, Adv. Math. Commun. 13 (1) (2019) 195–211.
\bibitem{A41}C. Zhu, Q. Liao, Complete weight enumerators for several classes of two-weight and three-weight linear codes, Finite Fields Appl. 75 (2021) 101897.
\bibitem{A46}C. Zhu, Q. Liao, Two new classes of projective two-weight linear codes, Finite fields and their applications. 88 (2023) 102186.
\bibitem{A44}D. Zheng, Q. Zhao, X. Wang, Y. Zhang, A class of two or three weights linear codes and their complete weight enumerators, Discrete Math. 344 (6) (2021) 112355.
\bibitem{A31}F. Özbudak, R. Pelen, Two or three weight linear codes from non-weakly regular bent functions, IEEE Trans. Inf. Theory 68 (5) (2022) 3014–3027.
\bibitem{A25}G. Jian, Z. Lin, R. Feng, Two-weight and three-weight linear codes based on Weil sums, Finite Fields Appl. 57 (2019) 92–107.
\bibitem{A28}G. Luo, X. Cao, A construction of linear codes and strongly regular graphs from $q$-polynomials, Discrete Math. 340 (2017) 2262–2274.
\bibitem{A30}H. Lu, S. Yang, Two classes of linear codes from Weil sums, IEEE Access. 8 (2020) 180471–180480.
\bibitem{A3}J. Yuan, C. Ding, Secret sharing schemes from three classes of linear codes, IEEE Trans. Inf. Theory 52 (1) (2006) 206–212.
\bibitem{A1}K. Torleiv, Codes for Error Detection (vol. 2), World Scientific, 2007.
\bibitem{A18}K. Ding, C. Ding, A class of two-weight and three-weight codes and their applications in secret sharing, IEEE Trans. Inf. Theory 61 (11) (2015) 5835–5842.
\bibitem{A60}M. Kiermaier,  S. Kurz, P. Solé , M. Stoll, On strongly walk regular graphs, triple sum sets and their codes, Designs, Codes and Cryptography. 91 (2023) 645-675.
\bibitem{A32}M. Qi, W. Xiong, Complete weight enumerators of two classes of linear codes with a few weights, Appl. Algebra Eng. Commun. Comput. 32 (2021) 63–79.
\bibitem{A62}M. Grassl, Bounds on the minimum distance of linear codes and quantum codes, Available: http://www.codetables.de.
\bibitem{A35}P. Tan, Z. Zhou, D. Tang, T. Helleseth, The weight distribution of a class of two-weight linear codes derived from Kloosterman sums, Cryptogr. Commun. 10 (2018) 291–299.
\bibitem{A47}R. Lidl and H. Niederreiter, Finite Fields, $2nd$ edition, Cambridge University Press, Cambridgeshire, 1997.
\bibitem{A38}S. Yang, Complete weight enumerators of a class of linear codes from Weil sums, IEEE Access. 8 (2020) 194631–194639.
\bibitem{A39}S. Yang, Q. Yue, Y. Wu, X. Kong, Complete weight enumerators of a class of two-weight linear codes, Cryptogr. Commun. 11 (2019) 609–620.
\bibitem{A57}T. Zhang, P. Ke, Z. Chang, Construction of three class of at most four-weight binary linear codes and their applications, submitted.
\bibitem{A6}W. C. Huffman and V. Pless, Fundamentals of Error-Correcting Codes, Cambridge University Press, Cambridgeshire, 2003.
\bibitem{A15}X. Cheng, X. Cao, L. Qian, Constructing few-weight linear codes and strongly regular graphs, Discrete Math. 345 (2022) 113101.
\bibitem{A36}X. Kong, S. Yang, X. Kong, Complete weight enumerators of a class of linear codes with two or three weights, Discrete Math. 342 (2019) 3166–3176.
\bibitem{A33}X. Quan, Q. Yue, X. Li, C. Li, Two classes of 2-weight and 3-weight linear codes in terms of Kloosterman sums, Cryptogr. Commun. (2022) 1–16.
\bibitem{A19}Z. Heng, D. Li, J. Du, F. Chen, A family of projective two-weight linear codes, Des. Codes Cryptogr. 89 (2021) 1993–2007.
\bibitem{A20}Z. Heng, Q. Yue, A class of binary linear codes with at most three weights, IEEE Commun. Lett. 19 (9) (2015) 1488–1491.
\bibitem{A23}Z. Heng, Q. Yue, A construction of q-ary linear codes with two weights, Finite Fields Appl. 48 (2017) 20–42.
\bibitem{A24}Z. Heng, Q. Yue, C. Li, Three classes of linear codes with two or three weights, Discrete Math. 339 (11) (2016) 2832–2847.
\bibitem{A22}Z. Heng, Q. Yue, Two classes of two-weight linear codes, Finite Fields Appl. 38 (2016) 72–92.



\end{thebibliography}
\end{document}